\documentclass[runningheads]{llncs}

\usepackage{latexsym,amssymb,amsmath}
\usepackage{xspace}
\usepackage{comment}
\usepackage{url}
\usepackage{enumerate}
\usepackage[defblank]{paralist}
\usepackage{graphicx}
\graphicspath{{imgs/}}
\usepackage{hyperref}
\usepackage{subfig}
\usepackage{microtype}
\usepackage{color}
\usepackage{todonotes}
\usepackage{times}

\newif\ifdraft
\draftfalse

\newcommand{\K}{\mathcal{K}} \renewcommand{\L}{\mathcal{L}}

\renewcommand{\S}{\mathcal{S}}

\newcommand{\mf}{\mathfrak}

\newcommand{\tup}[1]{\langle #1\rangle}            

\newcommand{\myi}{\emph{(i)}\xspace}
\newcommand{\myii}{\emph{(ii)}\xspace}
\newcommand{\myiii}{\emph{(iii)}\xspace}

\newcommand{\sys}{\ensuremath{\S}\xspace}

\newcommand{\mulpers}{\ensuremath{\muL_{P}}\xspace}

\newcommand{\muL}{\mu\L} 

\renewcommand{\mf}[1]{\Upsilon}

\title{Verification of Artifact-Centric Systems: \\Decidability and Modeling Issues}

\author{Dmitry Solomakhin\inst{1} \and Marco Montali\inst{1} \and Sergio Tessaris\inst{1} \and\\
 Riccardo De Masellis\inst{2}
}
\authorrunning{Dmitry Solomakhin et al.}

\institute{
 Free University of Bozen-Bolzano,
 Piazza Domenicani 3, 39100 Bolzano, Italy\\
 \email{(solomakhin|montali|tessaris)@inf.unibz.it}
 \and 
 Sapienza Universit\`a\  di Roma,
 Via Ariosto, 25, 00185 Rome, Italy\\
 \email{demasellis@dis.uniroma1.it}
}

\begin{document}

\maketitle
\sloppy

\setcounter{footnote}{0}

\begin{abstract}
Artifact-centric business processes have recently emerged as an approach in which processes are centred around the evolution of business entities, called \emph{artifacts}, giving equal importance to control-flow and data.
%
The recent Guard-State-Milestone (GSM) approach provides means for specifying business artifacts lifecycles in a declarative manner, using constructs that match how executive-level stakeholders think about their business.
However, it turns out that formal verification of GSM is undecidable even for very simple propositional temporal properties. We attack this challenging problem by translating GSM into a well-studied formal framework. We exploit this translation to isolate an interesting class of ``state-bounded'' GSM models for which verification of sophisticated temporal properties is decidable. We then introduce some guidelines to turn an arbitrary GSM model into a state-bounded, verifiable model.
\end{abstract}
\vspace*{-20pt}
\centerline{\small \textbf{Keywords:} artifact-centric systems, guard-stage-milestone, formal verification}

\section{Introduction}
\label{intro}

In the last decade, a plethora of graphical notations (such as BPMN and EPCs) have been proposed to capture business processes. Independently from the specific notation at hand, formal verification has been generally considered as a fundamental tool in the process design phase, supporting the modeler in building correct and trustworthy process models \cite{Mor2008}. Intuitively, formal verification amounts to check whether possible executions of the business process model satisfy some desired properties, like generic correctness criteria (such as deadlock freedom or executability of activities) or
 domain-dependent constraints.  To enable formal verification and other forms of reasoning support, the business process language gets translated into a corresponding formal representation, which typically relies on variants of Petri nets \cite{VDAS11}, transition systems \cite{ArP09}, or process algebras \cite{PW05}. Properties are then formalized using temporal logics,
using model checking techniques to actually carry out verification tasks  \cite{Clarke1999:ModelChecking}.

A common drawback of classical process modeling approaches is 
being \emph{activity-centric}: they mainly focus on the control-flow perspective, lacking the connection between the process and the data manipulated during its executions. This reflects also in the corresponding verification techniques, which often abstract away from the data component. 
%
%
%
%
This ``data and process engineering divide'' affects many contemporary process-aware information systems, incrementing the amount of redundancies and potential errors in the development phase \cite{Dum2011}. 
To tackle this problem, the artifact-centric paradigm has recently emerged as an approach in which processes are guided by the evolution of business data objects, called \emph{artifacts} \cite{Nigam03:artifacts,CH09}. A key aspect of artifacts is coupling the 
representation of data of interest, called \emph{information model}, with \emph{lifecycle constraints}, which specify the acceptable evolutions of the data maintained by the information model.
%
On the one hand, new modeling notations are being proposed to tackle artifact-centric processes. A notable example is the Guard-State-Milestone (GSM) graphical notation \cite{Damaggio:2011:EIF:2040283.2040315}, which corresponds to way executive-level stakeholders conceptualize their processes \cite{Bhatt-2007:artifacts-customer-engagements}.
On the other hand, formal foundations of the artifact-centric paradigm are being investigated in order to capture the relationship between processes and data and support formal verification  \cite{Deutsch:2009:AVD:1514894.1514924,BeLP12,DBLP:journals/corr/abs-1203-0024}. 
Two important issues arise in this setting. 
First, verification formalisms must go beyond propositional temporal logics, and incorporate first-order formulae to express constraints about the evolution of data and to query the information model of artifacts. 
Second, formal verification becomes much more difficult than for classical activity-centric approaches, even undecidable in the general case.

In this work, we tackle the problem of \emph{automated verification of GSM models}. 
First of all, we show that verifying GSM models is indeed a very challenging issue, being undecidable in general even for simple propositional reachability properties. We then provide a sound and complete encoding of GSM into Data-Centric Dynamic Systems (DCDSs), a recently developed formal framework for data- and artifact-centric processes \cite{DBLP:journals/corr/abs-1203-0024}. 
This encoding allows to reproduce in the GSM context the decidability and complexity results recently established for DCDSs with bounded information models (\emph{state-bounded DCDSs}). These are DCDSs where the number of tuples does not exceed a given maximum value. This does not mean that the system must contain an overall bounded number of data: along a run, infinitely many data can be encountered and stored into the information model, provided that they do not accumulate in the same state. 
We lift this property in the context of GSM, and show that verification of state-bounded GSM models is decidable for a powerful temporal logic, namely a variant of first-order $\mu$-calculus supporting a restricted form of quantification~\cite{Emerson96}. 
%
We then isolate an interesting class of GSM models for which state-boundedness is guaranteed, and introduce guidelines that can be employed to turn any GSM model into a state-bounded, verifiable model.

The rest of the paper is organized as follows. Section 2 gives an overview of GSM and provides a first undecidability result. Section 3 introduces DCDSs and presents the GSM-DCDS translation. Section 4 introduces ``state-bounded'' GSM models and provides key decidability results. Discussion and conclusion follow.

\section{GSM modeling of Artifact-Centric Systems}
The foundational character of artifact-centric business processes is the combination of static properties, i.e., the data of interest, and
dynamic properties of a business process, i.e., how it
evolves. \emph{Artifacts}, the key business entities of a given domain, are characterized by \myi an
\emph{information model} that captures business-relevant data, and
\myii a \emph{lifecycle model} that specifies how the artifact
progresses through the business.
In this work, we focus on the Guard-Stage-Milestone (GSM) approach for
artifact-centric modeling, recently proposed by IBM \cite{Damaggio:2011:EIF:2040283.2040315}.
\label{subsec-overview-gsm}
%
%
GSM is a declarative modeling framework that has been designed with the goal of being executable and at the same time enough high-level to result intuitive to executive-level stakeholders. 
The GSM information model uses (possibly nested) attribute/value pairs to capture the
domain of interest.
The key elements of a  lifecycle model are \emph{stages},
\emph{milestones} and \emph{guards}. 
Stages are (hierarchical) clusters of
activities (\emph{tasks}), intended to update and extend the
data of the information model. They are associated to milestones,
business operational objectives to be achieved when the stage is under execution. Guards control the activation of stages and, like milestones, are described in terms of data-aware expressions, called
\emph{sentries}, involving events and conditions over the artifact
information model. Sentries have the form $\textbf{on } e \textbf{ if }
cond$, where $e$ is an event and $cond$ is an (OCL-based) condition
over data. Both parts are optional, supporting pure event-based or
condition-based sentries.
Tasks represent the atomic units of work. Basic tasks are used to
update the information model of some artifact instance (e.g., by using
the data payload associated to an incoming event). Other tasks are
used to add/remove a nested tuple. A specific
\emph{create-artifact-instance} task is instead used to create a new
instance of a given artifact type; this is done by means of a two-way
service call, where the result is used to create a new tuple for the
artifact instance, assign a new identifier to it, and fill it with the
result's payload. Obviously, another task exists to remove a given
artifact instance.
In the following, we use \emph{model} for the intensional level of a specific
business process described in GSM, and \emph{instance} to denote a
GSM model with specific data for its information model.

The execution of a business process may involve several
\emph{instances} of artifact types described by a GSM model. At any
instant, the state of an artifact instance (\emph{snapshot}) is stored
in its information model, and is fully characterised by: \myi values of attributes in the data model, 
\myii status of its stages (open or closed) and
\myiii status of its milestones (achieved or invalidated).
%
Artifact instances may interact with the external world by exchanging
typed \emph{events}. In fact, \emph{tasks} are considered to be
performed by an external agent, and their corresponding execution is
captured with two event types: a \emph{service call},
whose instances are populated by the data from information model and then sent to the environment; 
and a \emph{service call return}, whose instances represent the
corresponding answer from the environment and are used to incorporate
the obtained result back into the artifact information model.
The environment can also send unsolicited (one-way) events, to trigger specific guards or milestones. 
Additionally, any change of a status attribute, such as opening a stage or  achieving a milestone, triggers an internal event, which can be further used to govern the artifact lifecycle.

\begin{example}
\label{ex:gsm}
\small
\begin{figure}[t]
\centering
\includegraphics[width=.8\textwidth]{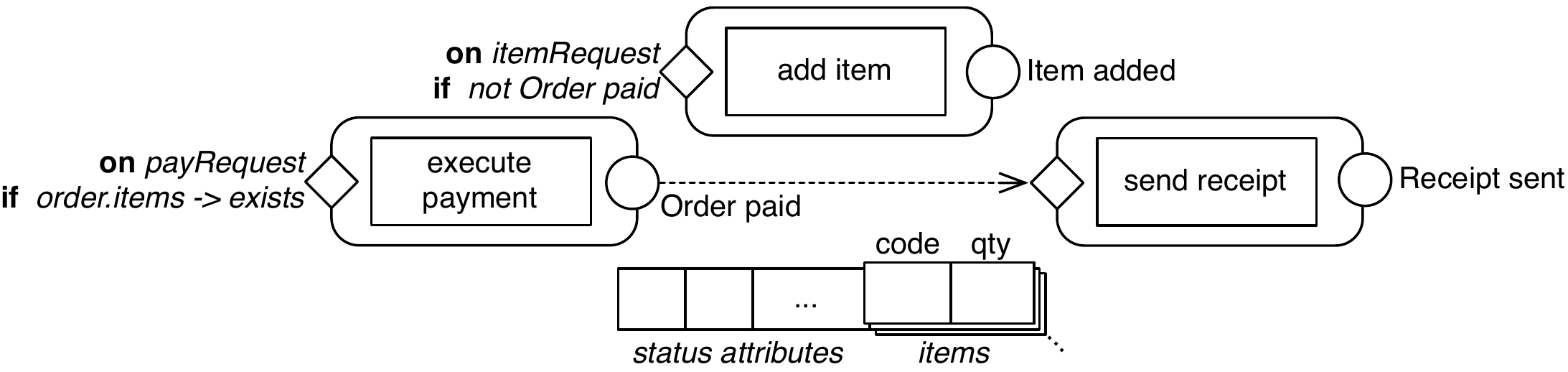}
\caption{GSM model of a simple order management process\label{fig:order-gsm}}
\end{figure}
Figure \ref{fig:order-gsm} shows a simple order management process
modeled in GSM. The process centers around an \emph{order} artifact, whose
information model is
characterized by a set of status attributes (tracking the status of
stages and milestones), and by an extendible set of ordered
\emph{items}, each constituted by a code and a quantity. The order
lifecycle contains three top-level atomic stages (rounded rectangles), respectively used to
manage the manipulation of the order, its payment, and the delivery of
a payment receipt. The order management stage contains a task
(rectangle) to add items to the order. It opens every time an $itemRequest$ event is received, 
provided that the order has not yet been paid. This is represented using a
logical condition associated to a guard (diamond). The stage closes when the task is executed, by achieving an ``item added''
milestone (circle). 
A payment can be
executed once a $payRequest$ event is issued, provided that the order
contains at least one item (verified by the OCL condition
$order.items \rightarrow exists$). As soon as the order is paid, and the
corresponding milestone achieved, the receipt delivery stage is
opened. This direct dependency is represented using a dashed arrow,
which is a shortcut for the condition $\textbf{on } Order~paid$, representing
the internal event of achieving the ``Order paid'' milestone. 
\end{example}
\subsection{Operational semantics of GSM}
\label{subsec-semantics-gsm}
%
%
GSM is associated to three well-defined, equivalent execution semantics, which
discipline the actual enactment of a GSM model \cite{Damaggio:2011:EIF:2040283.2040315}.
Among these, the \emph{GSM incremental semantics} is based on a form
of Event-Condition-Action (ECA) rules, called
\emph{Prerequisite-Antecedent-Consequent} (\emph{PAC}) rules, and is centered
around the notion of \emph{GSM Business steps} (\emph{B-steps}). An artifact instance remains idle until it receives an incoming event from the environment. It is assumed that such events arrive in a sequence and get processed by artifact instances one at a time. A B-step then describes what happens to an \emph{artifact snapshot} $\Sigma$, when a single incoming event $e$ is incorporated into it, i.e., how it evolves into a new snapshot $\Sigma'$ (see Figure 5 in \cite{Damaggio:2011:EIF:2040283.2040315}).
$\Sigma'$ is constructed by building a sequence of
pre-snapshots $\Sigma_i$, where $\Sigma_1$ results from incorporating $e$ into $\Sigma$ by updating its
attributes, one at a time, according to the event payload (i.e., its
carried data). Each consequent pre-snaphot $\Sigma_i$ is obtained by
applying one of the PAC rules to the previous pre-snapshot
$\Sigma_{i-1}$. Each of such transitions is called a \emph{micro-step}. During
a micro-step some outgoing events directed to the environment may be
generated. When no more PAC rules can be applied, the last
pre-snapshot $\Sigma'$ is returned, and the entire set of generated events is sent to the environment.

Each PAC rule is associated to one or more GSM constructs
(e.g. stage, milestone) and has three components:
\begin{compactitem}
\item\ \textbf{Prerequisite:} this component refers to the initial snapshot $\Sigma$ and determines if a rule is \emph{relevant} to the current B-step processing an incoming event $e$.
\item\ \textbf{Ancedent:} this part refers to the
  current pre-snapshot $\Sigma_i$ and determines whether the rule is
  eligible for execution, or $executable$, at the next micro-step.
\item\ \textbf{Consequent:} this part describes the effect of firing a rule, which can be nondeterministically chosen in order to obtain the next-pre-snapshot $\Sigma_{i+1}$.
\end{compactitem}
Due to nondeterminism in the choice of the next firing rule, different
orderings among the PAC rules can exist, leading to non-intuitive
outcomes. This is avoided in the GSM operational semantics by using an
approach reminiscent of stratification in logic programming. In
particular, the approach \myi exploits implicit dependencies between
the (structure of) PAC rules to fix an ordering on their execution,
and \myii applies the rules according to such ordering
\cite{Damaggio:2011:EIF:2040283.2040315}.
To guarantee B-step executability, avoiding situations in which the
execution indefinitely loops without reaching a stable state, the GSM
incremental semantics implements a so-called \emph{toggle-once}
principle. This guarantees that a sequence of micro-steps, triggered
by an incoming event, is always finite, by ensuring that each status
attribute can change its value at most once during a B-step. This
requirement is implemented by an additional condition in the prerequisite part of each PAC rule, which prevents it from firing twice.
%
%

The evolution of a GSM system composed by several artifacts can be
described by defining the initial state (initial snapshot of all
artifact instances) and the sequence of event instances generated by
the environment, each of which triggers a particular B-step, producing
a sequence of system snapshots. This perspective intuitively leads to
the representation of a GSM model as an infinite-state transition
system, depicting all possible sequences of snapshots supported by
the model. The initial configuration of the information model represents
the initial state of this transition system, and the incremental semantics provides the actual
transition relation. The source of infinity relies in the payload of
incoming events, used to populate the information model of artifacts
with fresh values (taken from an infinite/arbitrary domain). Since
such events are not under the control of the GSM model, the system must be
prepared to process such events in every possible order, and with
every acceptable configuration for the values carried in the payload.
The analogy to transition systems opens the possibility of using a
formal language, e.g., a (first-order variant of) temporal logic, to verify whether the GSM
system satisfies certain desired properties and requirements. For
example, one could test generic correctness properties, such as
checking whether each milestone can be achieved (and each stage will
be opened) in at least one of the
possible systems' execution, or that whenever a stage is opened, it
will be always possible to eventually achieve one of its
milestones. Furthermore, the modeler could also be interested in
verifying domain-specific properties, such as checking whether for the
GSM model in Figure~\ref{fig:order-gsm} it is possible to obtain a receipt before the payment is processed.

\subsection{Undecidability in GSM}
\label{subsec-undecidability-gsm}

\newcommand{\tm}{\ensuremath{\mathcal{M}}\xspace}
\newcommand{\dumsym}{\texttt{\#}}
\newcommand{\ktm}{\ensuremath{\K_\tm}\xspace}
\newcommand{\blank}{\textvisiblespace}

%


\begin{figure}[t!]
 \centering
   \includegraphics[width=\textwidth]{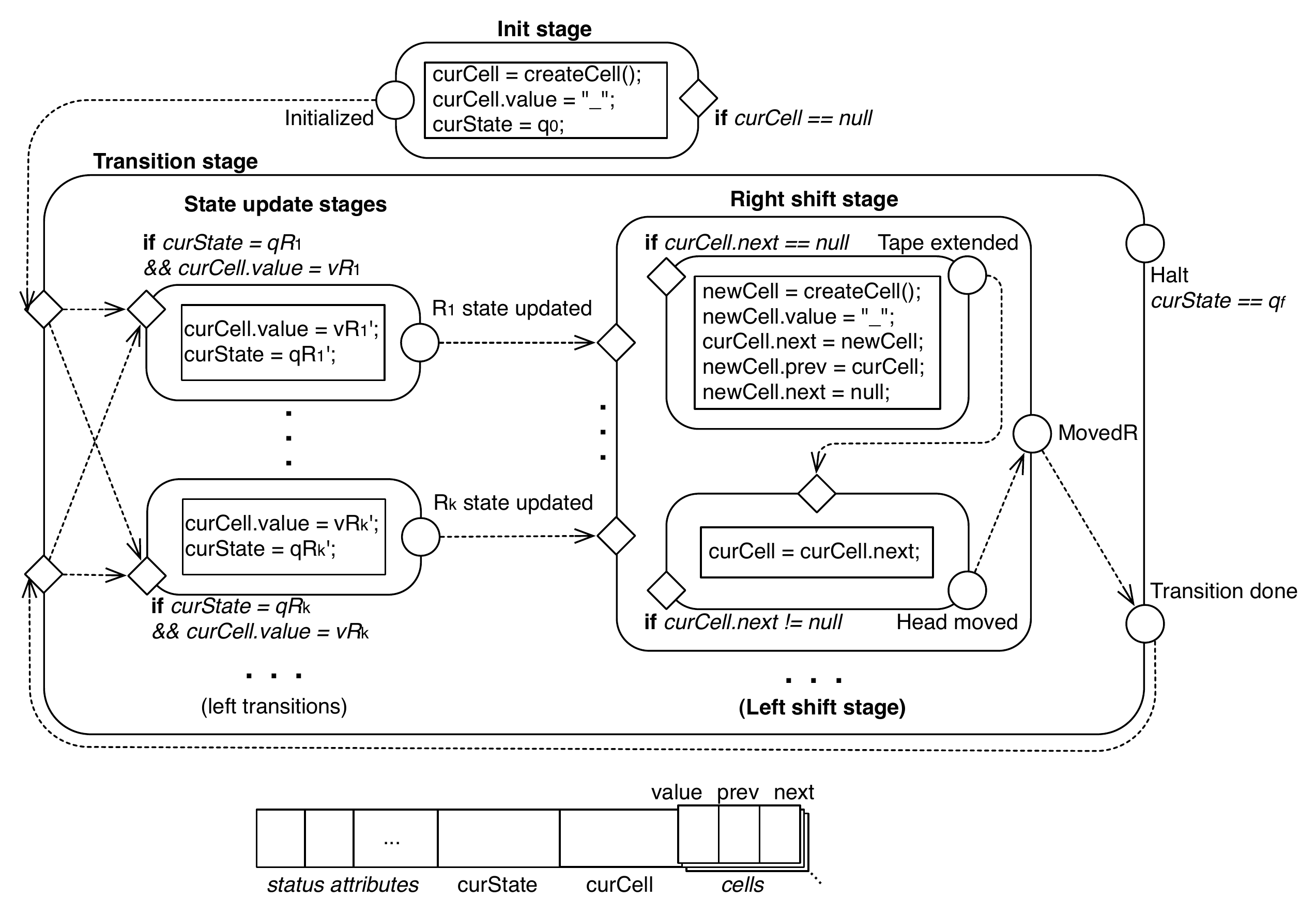}
\caption{ \label{fig:turing-gsm} GSM model of a Turing machine}
 \end{figure}

 In this section, we show that verifying the infinite-state transition
 system representing the execution semantics of a given GSM model is
 an extremely challenging problem, undecidable even for a very simple
 propositional reachability property.
\vspace*{-2pt}
\begin{theorem}
There exists a GSM model for which verification of a propositional
reachability property is undecidable.
\end{theorem}
\vspace*{-6pt}
\begin{proof}
To show undecidability of verification, we illustrate that a Turing machine can be
easily captured in GSM, and that the halting problem can be stated in
terms of a verification problem. In particular, we consider a
deterministic, single tape Turing machine $\tm =
\tup{Q,\Sigma,q_0,\delta,q_f,\blank}$, where  $Q$ is a finite set of
(internal) states, $\Sigma = \{0,1,\blank\}$ is the tape alphabet
(with $\blank$ the blank symbol), $q_0 \in Q$ and $q_f \in Q$ are the
initial and final state, and $\delta \subseteq Q\setminus \{q_f\}
\times \Sigma \times Q \times \Sigma \times \{L,R\}$ is a transition
relation.   We assume, wlog,  that $\delta$ consists of $k$ right-shift transitions
$R_1,\ldots,R_k$ (those having $R$ as last component), and $n$
left-shift transitions $L_1,\ldots,L_n$ (those having $L$ as last component). The
idea of translation into a GSM model is the following. Beside status
attributes, the GSM information model is constituted by:
\begin{inparaenum}[\it (i)]
\item a $curState$ slot containing the current internal state $q \in
  Q$;
\item a $curCell$ slot pointing to the cell where the head of $\tm$ is
  currently located.
\item a collection of $cells$ representing the current state of the tape.
\end{inparaenum}
Each cell is a complex nested record constituted by a value $v \in
\Sigma$, and two pointers $prev$ and $next$ used to link the cell to
the previous and next cells. In this way, the tape is modeled as a
linked list, which initially contains a single, blank cell, and which
is dynamically extended as needed. To mark the initial (resp., last)
cell of the tape, we assume that its $prev$ ($next$) cell is $null$.

On top of this information model, a GSM lifecyle that mimics $\tm$ is
shown in Figure~\ref{fig:turing-gsm}, where, due to space
constraints, only the right-shift transitions are depicted (the
left-shift ones are symmetric).
 The schema consists of two
top-level stages. \emph{Init} stage is used to initialize the tape.
\emph{Transition} stage is instead used to mimic the execution of one
of the transitions in $\delta$. Each transition is decomposed into two
sub-stages: \emph{state update} and \emph{head shift}. The state
update is modeled by one among $k+n$ atomic sub-stages, each handling
the update that corresponds to one of the transitions in
$\delta$. These stages are mutually exclusive, being $\tm$
deterministic. Consider for example a right-shift transition $R_i =
\delta(qR_i,vR_i,qR_i',vR_i',R)$ (the treatment is similar for a
left-shift transition). The corresponding state update stage is opened
whenever the current state is $qR_i$, and the value contained in the
cell pointed by the head is $vR_i$ (this can be extracted from the
information model using the query $curCell.value$). The incoming
arrows from the two parent's guards ensures that this condition is
evaluated as soon as the parent stage is opened; hence, if the
condition is true, the state update stage is immediately
executed. When the state update stage is closed, the achievement of
the corresponding milestone triggers one of the guards of the Right
shift stage that handles the head shift.
%
%
It contains two sub-stages: the first one extends the tape if the head
is currently pointing to the last cell, while the second one just
perform the shifting.
Whenever a right or left shift stage achieves the corresponding
milestone, then also the parent, transition stage is closed, achieving
milestone ``Transition done''. This has the effect of re-opening the
transition stage again, so as to evaluate the next transition to be
executed. An alternative way of immediately closing the transition
stage occurs when the current state corresponds to the final state
$q_f$. In this case, milestone ``Halt'' is achieved, and the execution
terminates (no further guards are triggered).

By considering this construction, the halting problem for
$\tm$ can be rephrased as the following verification problem: given
the GSM model encoding $\tm$, and starting from an initial state where the information model is empty, is it
possible to reach a state where the ``Halt'' milestone is achieved?
Notice that, since $\tm$ is deterministic, the B-steps of the
corresponding GSM model constitute a linear computation, which could
eventually reach the ``Halt'' milestone or continue
indefinitely. Therefore, reaching a state where ``Halt'' is achieved can
be equivalently formulated using propositional CTL or LTL. \qed
\end{proof}



\section{Translation into Data-Centric Dynamic Systems}
\label{sec-translation}
%
We discuss a translation procedure that faithfully rewrites a GSM model into a corresponding formal representation in terms of a Data-Centric Dynamic System (DCDS), for which interesting decidability results have been recently obtained.

DCDSs are a formal framework for the specification of data-aware business processes, i.e., systems where the connection between the process perspective and the manipulated data is explicitly tackled \cite{conf/dlog/HaririCGM11}. 
Technically, a DCDS is a pair $\mathcal{S} = \langle \mathcal{D}, \mathcal{P} \rangle$, where $\mathcal{D}$ is a data layer and $\mathcal{P}$ is a process layer over $\mathcal{D}$. $\mathcal{D}$ maintains all the relevant data in the form of a relational database with integrity constraints. In the artifact-centric context, the database is constituted by the union of all artifacts information models. The process layer $\mathcal{P}$ changes and evolves the data maintained by $\mathcal{D}$. It is constituted by a tuple 
$\mathcal{P} = \langle \mathcal{F}, \mathcal{A}, \varrho \rangle$. 
$\mathcal{F}$ is a finite set of functions representing interfaces to external services, used to import new, fresh data into the system.
$\mathcal{A}$ is a set of actions of the form $\alpha(p_1, ..., p_n) : \{e_1, ..., e_m\}$,  where $\alpha$ is the action name, $p_1, ..., p_n$ are input parameters, and $e_i$ are effect specifications. Each effect specification defines how a portion of the next database instance is constructed starting from the current one. Technically, its form is $Q \rightsquigarrow E$, where:
\begin{inparaenum}[\it (i)]
\item $Q$ is a query over $\mathcal{D}$ that could involve action parameters, and is meant to extract tuples from the current database;
\item $E$ is a set of effects, specified in terms of facts over $\mathcal{D}$ that will be asserted in the next state; these facts can contain variables of $Q$ (which are then replaced with actual values extracted from the current database), and also service calls, which are resolved by calling the service with actual input parameters and substituting them with the obtained result.\footnote{In \cite{conf/dlog/HaririCGM11}, two semantics for services are introduced: deterministic and nondeterministic. Here we always assume nondeterministic services, which is in line with GSM.}
\end{inparaenum}
Finally, $\varrho$ is a declarative process specified in terms of Condition-Action (CA) rules that determine, at any moment, which actions are executable.
Technically, each CA rule has the form $Q \mapsto \alpha$, where $Q$ is a query over $\mathcal{D}$, and $\alpha$ is an action. Whenever $Q$ has a positive answer over the current database, then $\alpha$ becomes executable, with actual values for its parameters given by the answer to $Q$. 

The execution semantics of a DCDS $\sys$ is defined by a possibly infinite-state transition system $\Upsilon_\sys$, where states are instances of the database schema in $\mathcal{D}$ and each transition corresponds to the application of an executable action in $\mathcal{P}$. Similarly to GSM, where the source of infinity comes from the fact that incoming events carry an arbitrary payload, in DCDSs the source of infinity relies in the service calls, which can inject arbitrary fresh values into the system.

We recall some key (un)decidability and complexity results related to DCDSs, which will be then used to study the formal verification of GSM.

\begin{theorem}[\cite{conf/dlog/HaririCGM11}]
There exists a DCDS for which verification of a propositional safety property expressible in LTL $\cap$ CTL  is undecidable.
\end{theorem}
This result comes from the high expressiveness of DCDSs. In fact, we will see that DCDSs can encode GSM. However, alongside this undecidability result, \cite{conf/dlog/HaririCGM11} identifies an interesting class of \emph{state-bounded} DCDSs, for which decidability of verification holds for a sophisticated (first-order) temporal logic called $\mulpers$. Intuitively, state boundedness requires the existence of an overall bound that limits, at every point in time, the size of the database instance of $\sys$ (without posing any restriction on which values can appear in the database). Equivalently, the size of each state contained in $\Upsilon_\sys$ cannot exceed the pre-established bound. Hence, in the following we will indifferently talk about state-bounded DCDSs or state-bounded transition systems.
\begin{theorem}[\cite{conf/dlog/HaririCGM11}]
\label{thm:decidability}
Verification of \mulpers properties over state-bounded DCDS is decidable, and can be reduced to finite-state model checking of propositional $\mu$-calculus.
\end{theorem}
$\mulpers$ is a first-order variant of $\mu$-calculus, a rich branching-time temporal logic that subsumes all well-known temporal logics such as PDL, CTL, LTL and CTL* \cite{Emerson96}.  $\mulpers$ employs first-order formulae to query data maintained by the DCDS data layer, and supports a controlled form of first-order quantification across states (within and across runs). In particular, $\mulpers$ requires that the values in the scope of quantification continuously persist for the quantification to take effect. As soon as a value is not present in the current database anymore, a formula talking about it collapses to $true$ or $false$. This restriction is in line with the artifact-centric setting, where a given artifact identifier points to the same artifact until such an artifact is live, but as soon as the artifact is destroyed, it can be recycled to identify a completely different artifact (and it would be incorrect to consider it the same as before). 
\vspace*{-2pt}
\begin{example}
\small
$\mulpers$ can express two variants of a correctness requirement for GSM:
\begin{compactitem}
\item it is always true that, whenever an artifact id is present in the information model, the corresponding artifact will be destroyed (i.e., the id will disappear) \emph{or} reach a state where all its stages are closed;
\item it is always true that, whenever an artifact id is present in the information model, the corresponding artifact will persist \emph{until} a state is reached where all its stages are closed.
\end{compactitem}
\end{example}


\subsection{Translating GSM into DCDS}
\label{subsec-gsm-dcds}
\newcommand{\gsm}{\ensuremath{\mathcal{G}}\xspace}
For the sake of space, we only discuss the intuition behind the translation and provide the main results. For a full technical development, we refer the interested reader to a technical report \cite{SMT2012}.

As introduced in Section~\ref{subsec-semantics-gsm}, the execution of a GSM instance is described by a sequence of B-steps. Each B-step consists of an initial micro-step which incorporates incoming event into current snapshot, a sequence of micro-steps executing all applicable PAC-rules,
and finally a micro-step sending a set of generated events at the termination of the B-step.
The translation relies on the incremental semantics: given a GSM model \gsm, we encode each possible micro-step as a separate condition-action rule in the process of a corresponding DCDS system \sys, such that the effect on the data and process layers of the action coincides with the effect of the corresponding micro-step in GSM. However, in order to guarantee that the transition system induced by a resulting DCDS mimics the one of the GSM model, the translation procedure should also ensure that all semantic requirements described in Section~\ref{subsec-semantics-gsm} are modeled properly: \begin{inparaenum}[\it (i)]
\item ``one-message-at-a-time'' and ``toggle-once'' principles,
\item the finiteness of micro-steps within a B-step, and
\item their order imposed by the model.
\end{inparaenum}
We sustain these requirements by introducing into the data layer of \sys a set of auxiliary relations, suitably recalling them in the CA-rules to reconstruct the desired behaviour.

Restricting \sys  to process only one incoming message at a time is implemented by the introduction of a \emph{blocking mechanism}, represented by an auxiliary relation $R_{block}(id_R, blocked)$ for each artifact in the system, where $id_R$ is the artifact instance identifier, and $blocked$ is a boolean flag. This flag is set to $true$ upon receiving an incoming message, and is then reset to $false$ at the termination of the corresponding B-step, once the outgoing events accumulated in the B-step are sent the environment. If an artifact instance has $blocked = true$, no further incoming event will be processed. This is enforced by checking the flag in the condition of each CA-rule associated to the artifact.

In order to ensure ``toggle once'' principle and guarantee the finiteness of sequence of micro-steps triggered by an incoming event, we introduce an \emph{eligibility tracking mechanism}. This mechanism is represented by an auxiliary relation $R_{exec}(id_R, x_1, ..., x_c)$,
 where $c$ is the total number of PAC-rules, and each $x_i$ corresponds to a certain PAC-rule of the GSM model. Each $x_i$ encodes whether the corresponding PAC rule is eligible to fire at a given moment in time (i.e., a particular micro-step). The initial setup of the eligibility tracking flags is performed at the beginning of a B-step, based on the evaluation of the prerequisite condition of each PAC rule. More specifically, when $x_i = 0$, the corresponding CA-rule is eligible to apply and has not yet been considered for application. When instead $x_i = 1$, then either the rule has been fired, or its prerequisite turned out to be false. This flag-based approach is used to propagate in a compact way information related to the PAC rules that have been already processed, following a mechanism that resembles \emph{dead path elimination} in BPEL. In fact, $R_{exec}$ is also used to enforce a firing order of CA-rules that follows the one induced by \gsm. This is achieved as follows. For each CA-rule $Q \mapsto \alpha$ corresponding to a given PAC rule $r$, condition $Q$ is put in conjunction with a further formula, used to check whether all the PAC rules that precede $r$ according to the ordering imposed by \gsm have been already processed. Only in this case $r$ can be considered for application, consequently applying its effect $\alpha$ to the current artifact snapshot. More specifically, the corresponding CA-rule becomes $Q \land exec(r) \mapsto \alpha$, where $exec(r) = \bigwedge_i x_i$ such that $i$ ranges over the indexes of those rules that precede $r$.

Once all $x_i$ flags are switched to $1$, the B-step is about to finish: a dedicated CA-rule is enabled to send the outgoing events to the environment, and the artifact instance $blocked$ flag is released.

 \begin{example}

\begin{figure}[t]
\scriptsize
\begin{align}
&R_{exec}(id_R,\overline{x}) \land x_k = 0 \land exec(k) \land R_{block}(id_R,true) \mapsto
\\
& \quad a_{exec}^k(id_R,\overline{a}',\overline{x}):\{\\
&\qquad R_{att}(id_R,\overline{a},\overline{s},\overline{m}) \land R_{chg}^{S_j}(id_R,true) \rightsquigarrow \{R_{att}(id_R,\overline{a},\overline{s},\overline{m})[m_j/false]\}\\
&\qquad R_{att}(id_R,\overline{a},\overline{s},\overline{m}) \land R_{chg}^{S_j}(id_R,true) \rightsquigarrow \{R_{chg}^{m_j}(id_R,false)\}\\
&\qquad R_{exec}^M(id_R,\overline{x}) \land x_k = 0 \rightsquigarrow \{R_{exec}^M(id_R,\overline{x})[x_k/1]\}\\
&\qquad [\mathsf{CopyMessagePools}],  [\mathsf{CopyRest}]\quad\}
\end{align}
\vspace*{-10pt}
\caption{ \label{fig:example-translation} CA-rule encoding a milestone invalidation upon stage activation}
\vspace*{-12pt}
\end{figure}
\small
An example of a translation of a GSM PAC-rule (indexed by $k$) is presented in Figure~\ref{fig:example-translation}.
For simplicity, multiple parameters are compacted using an ``array'' notation (e.g., $x_1,\ldots,x_n$ is denoted by $\overline{x}$). In particular: (1) represents a condition part of a CA-rule, ensuring the ``toggle-once'' principle ($x_k=0$), the compliant firing order ($exec(k)$) and the ``one-message-at-a-time'' principle ($R_{block}(id_R,true)$); 
(2) describes the action signature;
(3) is an effect encoding the invalidation a milestone once the stage has been activated; (4) propagates an internal event denoting the milestone invalidation, if needed; (5) flags the encoded micro-step corresponding to PAC rule $k$ as processed; (6) transports the unaffected data into the next snapshot.
\end{example}

\begin{figure}[t]
\centering
\includegraphics[width=\textwidth]{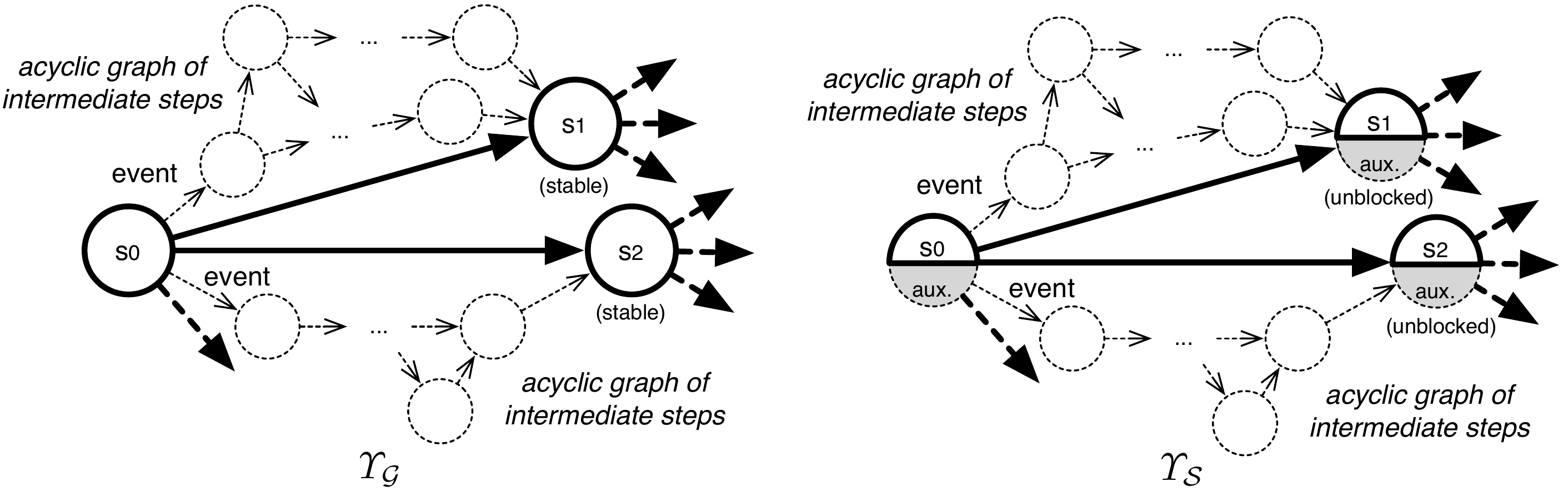}
\caption{Construction of the B-step transition system $\Upsilon_\gsm$ and unblocked-state transition system $\Upsilon_\sys$ for a GSM model \gsm with initial snapshot $s_0$ and the corresponding DCDS \sys}
\label{fig:ts}
\end{figure}

Given a GSM model \gsm with initial snapshot $S_0$, we denote by $\Upsilon_\gsm$ its \emph{B-step transition system}, i.e., the infinite-state transition system obtained by iteratively applying the incremental GSM semantics starting from $S_0$ and nondeterministically considering each possible incoming event. The states of $\Upsilon_\gsm$ correspond to stable snapshots of $\gsm$, and each transition corresponds to a B-step. We abstract away from the single micro-steps constituting a B-step, because they represent temporary intermediate states that are not interesting for verification purposes. Similarly, given the DCDS \sys obtained from the translation of \gsm, we denote by $\Upsilon_\sys$ its \emph{unblocked-state transition system}, obtained by starting from $S_0$, and iteratively applying nondeterministically the CA-rules of the process, and the corresponding actions, in all the possible ways. As for states, we only consider those database instances where all artifact instances are not blocked; these correspond in fact to stable snapshots of  $\gsm$.  We then connect two such states provided that there is a sequence of (intermediate) states that lead from the first to the second one, and for which at least one artifact instance is blocked; these sequence corresponds in fact to a series of intermediate-steps evolving the system from a stable state to another stable state. 
Finally, we project away all the auxiliary relations introduced by the translation mechanism, obtaining a \emph{filtered} version of $\Upsilon_\sys$, which we denote as $\Upsilon_\sys|_\gsm$. 
The intuition about the construction of these two transition systems is given in Figure~\ref{fig:ts}. Notice that the intermediate micro-steps in the two transition systems can be safely abstracted away because: \myi thanks to the toggle-once principle, they do not contain any ``internal'' cycle; \myii respecting the firing order imposed by \gsm, they all lead to reach the same next stable/unblocked state.
We can then establish the one-to-one correspondence between these two transition systems in the following theorem (refer to \cite{SMT2012} for complete proof):


\begin{theorem}
\label{thm:translation}
Given a GSM model \gsm and its translation into a corresponding DCDS \sys, the corresponding B-step transition system $\Upsilon_\gsm$ and filtered unblocked-state transition system $\Upsilon_\sys|_\gsm$ are equivalent, i.e., $\Upsilon_\gsm \equiv \Upsilon_\sys|_\gsm$.
\end{theorem}


\section{State-bounded GSM models}
\label{sec-bounded-gsm}


We now take advantage of the key decidability result given in Theorem~\ref{thm:decidability}, and study verifiability of \emph{state-bounded GSM models}.
Observe that state-boundedness is not a too restrictive condition. It requires each state of the transition system to contain a bounded number of tuples. However, this does not mean that the system in general is restricted to encounter only a limited amount of data: infinitely many values may be distributed \emph{across} the states (i.e. along an execution), provided that they do not accumulate in the same state. Furthermore, infinitely many executions are supported, 
reflecting that whenever an external event updates a slot of the information system maintained by a GSM artifact, infinitely many successor states in principle exist, each one corresponding to a specific new value for that slot.
To exploit this, we have first to show that the GSM-DCDS translation preserves state-boundedness, which is in fact the case.


\begin{lemma}
\label{lemma_statebounded}
\emph{Given a GSM model \gsm and its DCDS translation \sys, \gsm is state-bounded if and only if \sys is state-bounded.} 
\end{lemma}
\begin{proof}
Recall that \sys contains some auxiliary relations, used to restrict the applicability of CA-rules in order to enforce the execution assumptions of GSM: 
\begin{inparaenum}[\it (i)]
\item the eligibility tracking table $R_{exec}$,
\item the artifact instance blocking flags $R_{block}$,
\item the internal message pools $R^{msg_k}_{data}$, $R^{srv_p}_{data}$, $R^{msg_q}_{out}$, and
\item the tables of status changes $R^{m_i}_{chg}$, $R^{s_j}_{chg}$.
\end{inparaenum}
($\Leftarrow$) This is directly obtained by observing that, if $\Upsilon_\sys$ is state-bounded, then also $\Upsilon_\sys|_\gsm$ is state-bounded. From Theorem~\ref{thm:translation}, we know that $\Upsilon_\sys|_\gsm \equiv \Upsilon_\gsm$, and therefore $\Upsilon_\gsm$ is state-bounded as well.
\\
($\Rightarrow$) We have to show that state boundedness of \gsm implies that also all auxiliary relations present in $\Upsilon_\sys$ are bounded. We discuss each auxiliary relation separately. 
The artifact blocking relation $R_{block}$ keeps a boolean flag for each artifact instance, so its cardinality depends on the number of instances in the model. Since the model is state-bounded, the number of artifact instances is bounded and so is $R_{block}$.
The eligibility tracking table $R_{exec}$ stores for each artifact instance a boolean vector describing the applicability of a certain PAC rule. Since the number of instances is bounded and so is the set of PAC rules, then the relation $R_{exec}$ is also bounded. Similarly, one can show the boundedness of $R^{m_i}_{chg}$, $R^{s_j}_{chg}$ due to the fact that the number of stages and milestones is fixed a-priori. 
Let us now analyze internal message pools. By construction, \sys may contain at most one tuple in $R^{msg_k}_{data}$ and $R^{srv_p}_{data}$ for each artifact instance. This is enforced by the blocking mechanism $R_{block}$, which blocks the artifact instance at the beginning of a B-step  and prevents the instance from injecting further events in internal pools. The outgoing message pool $R^{msg_q}_{out}$ may contain as much tuples per artifact instance as the amount of atomic stages in the model, which is still bounded. However, neither incoming nor outgoing messages are accumulated in the internal pool along the B-steps execution, since the final micro-step of the B-step is designed not to propagate any of the internal message pools to the next snapshot. 
Therefore, $\Upsilon_\sys$ is state-bounded.


\qed
\end{proof}
From the combination of Theorems~\ref{thm:decidability} and~\ref{thm:translation} and Lemma~\ref{lemma_statebounded}, we directly obtain:
\begin{theorem}
\label{thm:gsm-dec}
Verification of \mulpers properties over state-bounded GSM models is decidable, and can be reduced to finite-state model checking of propositional $\mu$-calculus.
\end{theorem}
Obviously, in order to guarantee verifiability of a given GSM model, we need to understand whether it is state-bounded or not. However, state-boundedness is a ``semantic'' condition, which is undecidable to check \cite{DBLP:journals/corr/abs-1203-0024}. We mitigate this problem by isolating a class of GSM models that is guaranteed to be state-bounded. We show however that even very simple GSM models (such as Fig.~\ref{fig:order-gsm}), are not state-bounded, and thus we provide some modelling strategies to make any GSM model state-bounded. 


\subsubsection{GSM Models without Artifact Creation.}
\label{sec:no-creation}
We investigate the case of GSM models that do not contain any \emph{create-artifact-instance} tasks. Without loss of generality, we assimilate the creation of nested datatypes (such as those created by the ``add item'' task in Example~\ref{ex:gsm}) to the creation of new artifacts. From the formal point of view, we can in fact consider each nested datatype as a simple artifact with an empty lifecycle, and its own information model including a connection to its parent artifact.
\begin{corollary}
\label{cor-statebound}
Verification of \mulpers properties over GSM models without \emph{create-artifact-instance} tasks is decidable.
\end{corollary}
\begin{proof}
Let \gsm be a GSM model without \emph{create-artifact-instance} tasks. At each stable snapshot $\Sigma_k$, \gsm can either process an event representing an incoming one-way message, or the termination of a task. We claim that the only source of state-unboundedness can be caused by service calls return related to the termination of  \emph{create-artifact-instance} tasks.
In fact, one-way incoming messages, as well as other service call returns, do not increase the size of the data stored in the GSM information model, because the payload of such messages just substitutes the values of the corresponding data attributes, according to the signature of the message. Similarly, by an inspection of the proof of Lemma~\ref{lemma_statebounded}, we know that across the micro-steps of a B-step, status attributes are modified but their size does not change. Furthermore, a bounded number of outgoing events could be accumulated in the message pools, but this information is then flushed at the end of the B-step, thus bringing the size of the overall information model back to the same size present at the beginning of the B-step.
Therefore, without \emph{create-artifact-instance} tasks, the size of the information model in each stable state is constant, and corresponds to the size of the initial information model. We can then apply Theorem~\ref{thm:gsm-dec} to get the result.
%
\qed
\end{proof}



\subsubsection{Arbitrary GSM Models.}
\label{sec-guidelines}
\begin{figure}[t]
\includegraphics[width=.9\textwidth]{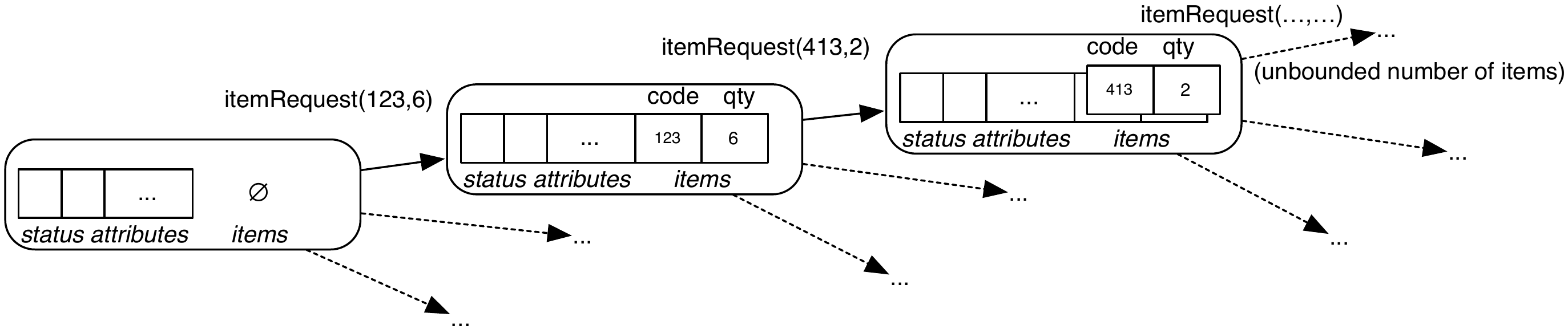}
\caption{\label{fig:unbounded-ts} Unbounded execution of the GSM model in Fig.~\ref{fig:order-gsm}}
\end{figure}
The types of models studied in paragraph above
 are quite restrictive, because they forbid the possibility of extending the number of artifacts during the execution of the system. On the other hand, as soon as this is allowed, even very simple GSM models, as the one shown in Fig.~\ref{fig:order-gsm}, may become state unbounded. In that example, the source of state unboundedness lies in the stage containing the ``add item'' task, which could be triggered an unbounded number of times due to continuous $itemRequest$ incoming events, as pointed out in Fig.~\ref{fig:unbounded-ts}. This, in turn, is caused by the fact that the modeler left the GSM model underspecified, without providing any hint about the maximum number of items that can be included in an order. To overcome this issue, we require the modeler to supply such information (stating, e.g., that each order is associated to at most $10$ items). Technically, the GSM model under study has to be parameterized by an arbitrary but finite number $N_{max}$, which denotes the maximum number of artifact instances that can coexist in the same execution state. We call this kind of GSM model \emph{instance bounded}.
A possible policy to provide such bound is to allocate available ``slots'' for each artifact type of the model, i.e. to specify a maximum number $N_{A_i}$ for each artifact type $A_i$, then having $N_{max} = \sum_i N_{A_i}$.
%
In order to incorporate the artifact bounds into the execution semantics, we proceed as follows. First, we pre-populate the initial snapshot of the considered GSM instance with $N_{max}$ blank artifact instances (respecting the relative proportion given by the local maximum numbers for each artifact type). We refer to one such blank artifact instance as \emph{artifact container}. Along the system execution, each container may be:
\begin{inparaenum}[\it (i)]
\item filled with concrete data carried by an actual artifact instance of the corresponding type, or
\item flushed to the initial, blank state. 
\end{inparaenum}
To this end, each artifact container is equipped with an auxiliary flag $fr_i$, which reflects its current state: $fr_i$ is false when the container stores a concrete artifact instance, true otherwise. 
Then, the internal semantics of \emph{create-artifact-instance} is changed so as to check the availability of a blank artifact container. In particular, when the corresponding service call is to be invoked with the new artifact instance data, the calling artifact instance selects the next available blank artifact container, sets its flag $fr_i $ to $false$, and fills it with the payload of the service call. If all containers are occupied, the calling artifact instance waits until some container is released. 
Symmetrically to artifact creation, the deletion procedure for an artifact instance is managed by turning the corresponding container flag $fr_i$ to true. Details on the DCDS CA-rules formalizing creation/deletion of artifact instances according to these principles can be found in \cite{SMT2012}.

We observe that, following this container-based realization strategy, the information model of an instance-bounded GSM model has a fixed size, which polinomially depends on the total maximum number $N_{max}$. The new implementation of \emph{create-artifact-instance} does not really change the size of the information model, but just suitably changes its content. Therefore, Corollary~\ref{cor-statebound} directly applies to instance-bounded GSM models, guaranteeing decidability of their verification.  Finally, notice that infinitely many different artifact instances can be created and manipulated, provided that they do not accumulate in the same state (exceeding $N_{max}$). 



\section{Discussion and related work}
\label{sec-discussion}
%
In this work we have provided the foundations for the formal verification of the GSM artifact-centric paradigm.
After having proven undecidability of verification in the general case, we have shown decidability of verification for a very rich first-order temporal logic, tailored to the artifact-centric setting, for an interesting class of ``state-bounded'' GSM models.

So far, only few works have investigated verification of GSM models. The closest approach to ours is \cite{BeLP12b}, where state-boundedness is also used as a key property towards decidability. The main difference between the two approaches is that decidability of state-bounded GSM models is proven for temporal logics of incomparable expressive power. In addition to \cite{BeLP12b}, in this work we also study modeling strategies to make an arbitrary GSM model state-bounded, while they assume that the input model is guaranteed to be state-bounded. Hence, our strategies could be instrumental to \cite{BeLP12b} as well.
In \cite{GoGL12}, another promising technique for the formal verification of GSM models is presented. However, the current implementation
cannot be applied to general GSM models, because of assumptions over the data types and the fact that only one instance per artifact type is supported. Furthermore, a propositional branching-time logic is used for verification, restricting to the status attributes of the artifacts. The results presented in our paper can be used to generalize this approach towards more complex models (such as instance-bounded GSM models) and more expressive logics, given, e.g., the fact that ``one-instance artifacts'' fall inside the decidable cases we discussed in this paper.

It is worth noting that all the presented decidability results are actually even stronger: they state that verification can be reduced to standard model checking of propositional $\mu$-calculus over finite-state transition systems (thanks to the abstraction techniques studied in \cite{DBLP:journals/corr/abs-1203-0024}). This opens the possibility of actually implementing the discussed techniques, by relying on state-of-the-art model checkers. We also inherit from \cite{DBLP:journals/corr/abs-1203-0024} the complexity boundaries: they state that verification is \textsc{ExpTime} in the size of the GSM model which, in the case of instance-bounded GSM models, means in turn \textsc{ExpTime} in the maximum number of artifact instances that can coexist in the same state.

Beside implementation-related issues, we also aim to reassess the results presented here in a setting where GSM relies on a rich knowledge base  (a description logic ontology) for its information model, in the spirit of \cite{BCG12}.

\vspace{-.3cm}
\bibliographystyle{splncs03}
\bibliography{main-bib}

\newcommand{\SortNoOp}[1]{}
\begin{thebibliography}{10}
\providecommand{\url}[1]{\texttt{#1}}
\providecommand{\urlprefix}{URL }

\bibitem{VDAS11}
van~der Aalst, W.M.P., Stahl, C.: Modeling Business Processes - A Petri
  Net-Oriented Approach. Springer (2011)

\bibitem{ArP09}
Armando, A., Ponta, S.E.: Model checking of security-sensitive business
  processes. In: Proc.\ of FAST. LNCS, vol. 5983, pp. 66--80. Springer (2009)

\bibitem{conf/dlog/HaririCGM11}
Bagheri~Hariri, B., Calvanese, D., De~Giacomo, G., De~Masellis, R.:
  Verification of conjunctive-query based semantic artifacts. In: Proceedings
  of the 24th International Workshop on Description Logics (DL 2011). CEUR
  Workshop Proceedings, vol. 745. CEUR-WS.org (2011)

\bibitem{BCG12}
Bagheri~Hariri, B., Calvanese, D., Giacomo, G.D., Masellis, R.D., Felli, P.,
  Montali, M.: Verification of description logic knowledge and action bases.
  In: Proc.\ of ECAI. vol. 242, pp. 103--108. IOS Press (2012)

\bibitem{BeLP12}
Belardinelli, F., Lomuscio, A., Patrizi, F.: An abstraction technique for the
  verification of artifact-centric systems. In: Proc. of KR. AAAI Press (2012)

\bibitem{BeLP12b}
Belardinelli, F., Lomuscio, A., Patrizi, F.: Verification of gsm-based
  artifact-centric systems through finite abstraction. In: Proc.\ of ICSOC.
  LNCS, vol. 7636, pp. 17--31. Springer (2012)

\bibitem{Bhatt-2007:artifacts-customer-engagements}
Bhattacharya, K., Caswell, N.S., Kumaran, S., Nigam, A., Wu, F.Y.:
  Artifact-centered operational modeling: {Lessons} from customer engagements.
  IBM Systems Journal  46(4),  703--721 (2007)

\bibitem{Clarke1999:ModelChecking}
Clarke, E.M., Grumberg, O., Peled, D.A.: Model checking. The MIT Press (1999)

\bibitem{CH09}
Cohn, D., Hull, R.: Business artifacts: A data-centric approach to modeling
  business operations and processes. IEEE Data Eng. Bull.  32(3) (2009)

\bibitem{Damaggio:2011:EIF:2040283.2040315}
Damaggio, E., Hull, R., Vaculin, R.: On the equivalence of incremental and
  fixpoint semantics for business artifacts with guard-stage-milestone
  lifecycles. Information Systems  (2012)

\bibitem{Deutsch:2009:AVD:1514894.1514924}
Deutsch, A., Hull, R., Patrizi, F., Vianu, V.: Automatic verification of
  data-centric business processes. In: Proc.\ of ICDT. pp. 252--267. ICDT '09,
  ACM (2009)

\bibitem{Dum2011}
Dumas, M.: On the convergence of data and process engineering. In: Eder, J.,
  Bielikov{\'a}, M., Tjoa, A.M. (eds.) ADBIS. LNCS, vol. 6909, pp. 19--26.
  Springer (2011)

\bibitem{Emerson96}
Emerson, E.A.: Model checking and the mu-calculus. In: Descriptive Complexity
  and Finite Models (1996)

\bibitem{GoGL12}
Gonzalez, P., Griesmayer, A., Lomuscio, A.: Verifying gsm-based business
  artifacts. In: Proc.\ of ICWS. pp. 25--32. IEEE (2012)

\bibitem{DBLP:journals/corr/abs-1203-0024}
Hariri, B.B., Calvanese, D., Giacomo, G.D., Deutsch, A., Montali, M.:
  Verification of relational data-centric dynamic systems with external
  services. CoRR  abs/1203.0024 (2012)

\bibitem{Mor2008}
Morimoto, S.: A survey of formal verification for business process modeling.
  In: Computational Science (ICCS 2008), LNCS, vol. 5102, pp. 514--522.
  Springer (2008)

\bibitem{Nigam03:artifacts}
Nigam, A., Caswell, N.S.: Business artifacts: An approach to operational
  specification. IBM Systems Journal  42(3) (2003)

\bibitem{PW05}
Puhlmann, F., Weske, M.: Using the {\it pi}-calculus for formalizing workflow
  patterns. In: Proceedings of the 3rd International Conference on Business
  Process Management. vol. 3649, pp. 153--168 (2005)

\bibitem{SMT2012}
Solomakhin, D., Montali, M., Tessaris, S.: Formalizing guard-stage-milestone
  meta-models as data-centric dynamic systems. Tech. Rep. KRDB12-4, KRDB
  Research Centre, Faculty of Computer Science, Free University of
  Bozen-Bolzano (2012)

\end{thebibliography}

\end{document}